\documentclass{article}
\usepackage{geometry}
\usepackage{natbib}
\usepackage{times}
\usepackage{soul}
\usepackage{url}
\usepackage[hidelinks]{hyperref}
\usepackage[utf8]{inputenc}
\usepackage[small]{caption}
\usepackage{graphicx}
\usepackage{booktabs}
\usepackage{xcolor}
\usepackage{amssymb}
\usepackage{physics}

\usepackage{booktabs}
\usepackage{amsthm}
\usepackage{amsmath}
\usepackage{soul}
\usepackage{algorithm}
\usepackage{algpseudocode}
\usepackage{subcaption}
\usepackage{amsfonts}
\usepackage{cleveref}

\usepackage{multirow}

\title{Bidding in Smart Grid PDAs: Theory, Analysis and Strategy\break
\Large (Extended Version)}

\author{Susobhan Ghosh,\textsuperscript{1} Sujit Gujar,\textsuperscript{1} Praveen Paruchuri,\textsuperscript{1} Easwar Subramanian,\textsuperscript{2} Sanjay P. Bhat\textsuperscript{2}\\
\textsuperscript{1} Machine Learning Lab, IIIT Hyderabad, India \\
\textsuperscript{2} Tata Consultancy Services, TCS Innovation Labs, Hyderabad, India\\
\texttt{\small\{susobhan.ghosh@research.iiit.ac.in\}, \{sujit.gujar, praveen.p\}@iiit.ac.in}\\ \texttt{\small\{easwar.subramanian, sanjay.bhat\}@tcs.com}}

\theoremstyle{plain}
\newtheorem{theorem}{Theorem}

\theoremstyle{definition}

\newtheorem{definition}{Definition}

\newtheorem{claim}{Claim}
\newtheorem{proposition}{Proposition}
\newtheoremstyle{case}{}{}{}{}{\bfseries}{:}{ }{}
\theoremstyle{case}
\newtheorem{case}{Case}

\begin{document}

\maketitle

\begin{abstract}
Periodic Double Auctions (PDAs) are commonly used in the real world for trading, e.g. in stock markets to determine stock opening prices, and energy markets to trade energy in order to balance net demand in smart grids, involving trillions of dollars in the process. A bidder, participating in such PDAs, has to plan for bids in the current auction as well as for the future auctions, which highlights the necessity of good bidding strategies. In this paper, we perform an equilibrium analysis of single unit single-shot double auctions with a certain clearing price and payment rule, which we refer to as ACPR, and find it intractable to analyze as number of participating agents increase. We further derive the best response for a bidder with complete information in a single-shot double auction with ACPR. Leveraging the theory developed for single-shot double auction and  taking the PowerTAC wholesale market PDA as our testbed, we proceed by modeling the PDA of PowerTAC as an MDP. We propose a novel bidding strategy, namely MDPLCPBS. We empirically show that MDPLCPBS follows the equilibrium strategy for double auctions that we previously analyze. In addition, we benchmark our strategy against the baseline and the state-of-the-art bidding strategies for the PowerTAC wholesale market PDAs, and show that MDPLCPBS outperforms most of them consistently.
\end{abstract}

\section{Introduction}
% \sg{first two paras can be shrunk}
Auctions are mechanisms which facilitate buying and selling of goods or items amongst a group of agents. Double auctions are prevalent when both the sides of a market actively bid. For example, in the New York Stock Exchange, opening prices are determined using double auctions \cite{parsons2011auctions}. In smart grids, multiple power generating companies and different distributing agencies (brokers) trade electricity in the wholesale markets using double auctions.

In this work, we focus primarily on electricity markets. In July 2019, approximately 1.2 Billion Euros worth electricity was traded in Nord Pool alone, with 52\% of the volume being traded using APIs \cite{nordpool}. Any small improvement in cost optimization by deploying better bidding strategies can lead to significant improvements in the profits of the distributing agencies. Motivated by this, we take up a formal game-theoretic approach in this work for devising bidding strategies. 

Typically, for double auctions, clearing price and payment rules differ from market to market. Equilibrium analysis of double auctions has been explored extensively with different payment and clearing price rules \cite{wilson1992strategic}. Specifically, for $k$-double auctions, \citeauthor{satterthwaite1989bilateral} (\citeyear{satterthwaite1989bilateral}) proved the existence of multiple non-trivial equilibria for $k \in [0, 1]$. They also focused on a class of well-behaved equilibria, by making generalist assumptions on buyer's and seller's bidding strategies. Our focus in this paper is \emph{Average Clearing Price Rule (ACPR) based Periodic Double Auctions (PDAs)}, commonly used in smart grids (Power TAC \cite{ketter20172018}). In ACPR, the clearing price set as the average of last executing bid and last executing ask (a special case of $k$-double auction with $k$ = 0.5).

For ACPR, \citeauthor{chatterjee1983bargaining} (\citeyear{chatterjee1983bargaining}) constructed a symmetric equilibrium for the case of one buyer and one seller with uniformly distributed valuations. However, in the vast literature of double auctions, a generic equilibrium analysis for ACPR with more buyers has not been well explored \cite{wilson1992strategic}. We take up a double auction with ACPR as a case study. We assume all the agents involved (buyers and sellers) deploy scaling based strategies, and identify the Nash Equilibrium (NE) of the induced game. Researchers have used fictitious play-based convergence to equilibrium (e.g., \cite{shi2010equilibrium}) in double auctions. However, such strategies are not useful in PDAs when the agents need to place bids in real-time for new auctions. In such settings, we believe scaling based strategies are easy to interpret and implement. The equilibrium analysis of non-linear or other complex forms are analytically difficult to compute; moreover may not be appealing to the real users of these markets. 

We characterize NEs for One Buyer and One Seller (OBOS) and Two Buyer and One Seller (TBOS) analytically (Theorem \ref{OBOSThm} and \ref{TBOSThm}). Given our assumption of scaling based bidding strategies and uniform type distributions, generic equilibrium analysis of double auctions, following ACPR, beyond these settings is challenging. %\sg{there should be a connection with the next part}
To test our double auction strategies, we take the help of the PowerTAC simulation environment. PowerTAC is a simulation platform that replicates crucial elements of the smart grid, where multiple distributing agencies (brokers) compete across markets to generate the most profit. Note that, the double auctions in PowerTAC; for that matter actually in electricity markets; are PDAs. In PDAs, the market clears multiple times, each after a specific time interval.

% \sg{see if u can reduce this paragraph...}
Now, if a buyer knows all the bids in a double auction, we argue that it is a best response for the buyer to bid as close as possible to the last clearing bid in order to procure the full required energy (Proposition \ref{prop:LCP}). However, in reality, buyers never have access to such information. To address this incomplete information, we model the bidding process in PowerTAC PDAs as a Markov Decision Process (MDP), and solve it using dynamic programming and Last Clearing Price (LCP) prediction. Motivated by Power TAC's fast response time constraints, we propose a PDA bidding strategy \emph{MDPLCPBS} (Algorithm \ref{MDPLCPBS}). Though our MDP formulation is inspired by \citeauthor{urieli2014tactex} (\citeyear{urieli2014tactex}), the novelty lies in the reward, solution, and application to place bids. First, we illustrate that the MDP based strategy indeed achieves the equilibrium strategy characterized for OBOS setting. Then, we conduct different experiments to compare MDPLCPBS with the following strategies: ZI \cite{gode1993allocative}, ZIP \cite{tesauro2001high}, TacTex \cite{urieli2014tactex}, and MCTS \cite{chowdhury2018bidding}. Our analysis shows that MDPLCPBS outperforms ZI, TacTex, and ZIP in all the cases, and closely matches with MCTS. Simultaneously, we show that it predicts the LCP with minimal error. We used this bidding strategy to great effect during PowerTAC 2018 Finals \cite{passghosh} \cite{Ghosh2019}.

% \sg{to me, the rest of para is not required here...we can put this after we discuss our experimental analysis}Note that in these experiments, as the energy to be procured is same across all the brokers, we set it as some proportion of the net demand in the PowerTAC simulation tariff market. MCTS is a heuristic-based bidding strategy, whereas MDPLCPBS is based on the game-theoretic analysis of a single shot double auction. We leave the future work to analyze MCTS in a game-theoretic framework.

In summary, our contributions are as follows:
\begin{itemize}
    \item We analytically characterize NE strategies for OBOS and TBOS settings (Theorem \ref{OBOSThm} and Theorem \ref{TBOSThm}).
    \item We propose the best response in a complete information multi-unit double auction.
    \item For bidding in PDAs such as PowerTAC, we design an algorithm MDPLCPBS (Algorithm \ref{MDPLCPBS}). It is based on dynamic programming and LCP prediction.
    \item Experimentally, we validate that MDPLCPBS achieves the equilibrium characterized for OBOS setting. Further, we demonstrate its efficacy against state of the art strategies for PowerTAC, and also show that it predicts the LCP with minimal error.
\end{itemize}{}

\section{Definitions \& Background}
% \sg{We first define all the required terms formally. }
% \sg{first we shud define pda, k double auciton, acpr then best response ne n then mdp..intro talks about these in that order}
% Nash Equilibrium, Best response definition from Narahari. MDP, PDA Definition, PowerTAC setting, and WM description.
We first define all the required terms formally.

\begin{definition}(\emph{Periodic Double Auction (PDA)}) A type of auction, for buying and selling some resource, with multiple discrete clearing periods i.e. clearing after a specific time interval. Potential buyers submit their bids and potential sellers simultaneously submit their asks to an auctioneer. Then the auctioneer matches the bids and asks, and chooses some \emph{clearing price}, denoted as $CP$, that clears the auction \cite{wurman1998flexible}. The \emph{allocation rule} determines the quantity bought/sold by each buyer/seller, while the \emph{payment rule} determines how much each buyer/seller pays/earns for buying/selling that quantity.
\end{definition}

\begin{definition}(\emph{Last Clearing Bid/Ask (LCB/LCA)}) Last Clearing Bid (Ask) of an auction refers to that partially or fully cleared bid (ask) which has the lowest (highest) limit-price. It is referred to as ``last clearing" since it is the last bid (ask) to be cleared by the clearing mechanism of the auction.
\end{definition}

\begin{definition}(\emph{Last Clearing Price (LCP)}) Last Clearing Price (LCP) of bids (asks) refers to the limit-price of the Last Clearing Bid (Ask).
\end{definition}

\begin{definition}(\emph{The k-Double Auction}) If a buyer and seller participate in a double auction, and if the sealed bid $b$ by the buyer is higher than the sealed bid $s$ by the seller, then $CP$ is given by $kb + (1-k)s$ for some fixed $k \in [0,1]$. % \sg{is it open interval or closed...equlibrium analysis is for open internval accoring to intro}
\end{definition}

\begin{definition}(\emph{Average Clearing Price Rule (ACPR)})
In a double auction, the clearing price and payment rule is \emph{ACPR} if the clearing price is given by $(b + s)/2$ where $b$ is the last executed bid, and $s$ is the last executed ask. It is a special case of k-double auction, with $k = 0.5$.
\end{definition}

\sloppy Consider a game $\Gamma$ $=$ $\langle N, (S_i)_{i \in N}$, $(u_i)_{i \in N}\rangle$, where $N = \{1, 2, \ldots, n\}$ is the set of players, $S_i$ is the strategy set of the player $i$, and $u_i : S_1 \times S_2 \times \ldots \times S_n \to \mathbb{R}$ for $i = 1, 2, \ldots n$ are utility functions.

\begin{definition}(\emph{Best Response})
\sloppy Given a game $\Gamma$, the best response correspondence for player $i$ is the mapping $B_i: S_{-i} \to S_i$ defined by $B_i(s_{-i}) = \{s_i \in S_i: u_i(s_i, s_{-i}) \geq u_i(s_i', s_{-i}) \forall s_i' \in S_i\}$. That is, given a profile $s_{-i}$ of strategies of the other players, $B_i(s_{-i})$ gives the set of all best response strategies of player $i$. %\cite{narahari2014game}. Add this in the final version
\end{definition}

\begin{definition}(\emph{Nash Equilibrium})
\sloppy Given a game $\Gamma$, a strategy profile $s^{*} = (s^*_1, s^*_2, \ldots, s^*_n)$ is said to be a Nash Equilibrium of $\Gamma$ if, $u_i(s^{*}_i, s^{*}_{-i}) \geq u_i(s_i, s^{*}_{-i}) \forall s_i \in S_i, \forall i = 1, 2, \ldots, n$. That is, each player's Nash Equilibrium strategy is a best response to the Nash Equilibrium strategies of the other players. %\cite{narahari2014game}.
\end{definition}

\begin{definition}(\emph{Markov Decision Process (MDP)} \cite{puterman1994markov})
A Markov Decision Process (MDP) is a tuple given by $M = (S, A, P, r, \gamma)$ where $S$ is the set of states, $A$ is the set of actions, $P$ is the state transition probability function, where $P(s'|s,a) = P(s_{t+1} = s'|s_t = s, a_t = a)$ is the probability that action $a$ in state $s$ at time $t$ will lead to state $s'$ at time $t+1$, $r$ is the reward function, with $r(s,a)$ denoting the reward obtained by taking action $a$ in state $s$, and $\gamma \in [0, 1]$ is the discount factor.
\end{definition}

% \sg{in general powertac you mat wnt to give the latest reference}
\paragraph{PowerTAC} In this work, we focus more on smart grids. The Power Trading Agent Competition (PowerTAC) \cite{ketter20172018} environment simulates a smart grid for approximately 60 days, where multiple brokers compete against each other across three markets - tariff, wholesale and balancing market - to generate the most profit. Each broker maintains a portfolio of consumers and producers, and buys and sells energy in the wholesale market. The broker with the highest bank balance at the end of the simulation, wins the game. We use the PowerTAC simulator to benchmark our bidding strategy. 

The PowerTAC wholesale market employs PDAs for wholesale market energy trading. The \emph{clearing price and payment rule} for the PowerTAC PDA, is given by $ACPR$. Three types of entities participate in these auctions - (1) Generating Companies (GenCos), (2) Miso Buyer, and (3) PowerTAC brokers. GenCos place only asks to sell energy, while the Meso Buyer places very low bid prices to buy energy. The PowerTAC brokers are free to place a bid or an ask depending on their requirement, or not place a bid at all. 

% \sg{you may want to reduce the desciption about powertac}
The brokers can always participate in 24 auctions to trade energy, one auction for each of the next 24 timeslots. Each broker is notified about identity of other brokers participating in the PDAs at the beginning of the simulation. Each broker estimates its own energy requirement, and knows its own type. However, it does not know the types and requirements of the competing brokers. Every broker is allowed to submit an unlimited number of bids for each auction. After clearance, the clearing price and total cleared quantity of the auction is made public to all the brokers, while the last cleared bid or ask is not revealed. Additionally, each broker is privately notified about the cleared quantity and clearing price of any of its cleared bids/asks. The orderbook of the auction, which is the set of uncleared bids and asks without identity of the bidders, is also made public to all the brokers. If a broker fails to balance its retail demand portfolio after all the 24 auctions in the wholesale market, the balancing market automatically supplies the energy while charging the broker a $balancing\text{-}price$ for its imbalance. The $balancing\text{-}price$ is comparatively higher than the wholesale market price, and is meant to penalize the broker for having an imbalance. For more details about the PowerTAC simulation, we refer the reader to the Power TAC 2018 Game Specification \cite{ketter20172018}.

\section{Related Work}
% \sg{if some these are talked in intro..u may shrink it}
Most bidding strategies for double auctions are designed for \emph{Continuous Double Auctions} (CDAs) and would need to be modified for PDAs. Bidding strategies for PDAs, outside PowerTAC, are very limited. \citeauthor{wah2016strategic} (\citeyear{wah2016strategic}) showed that in equilibrium, slow traders have higher welfare compared to fast traders in PDAs. As for bidding strategies for the PowerTAC wholesale market, AstonTAC \cite{kuate2013intelligent} uses Non-Homogeneous Hidden Markov Models (NHHMM) to predict energy demand and clearing price, which are then fed to an MDP to determine bid prices. TacTex \cite{urieli2014tactex,urieli2016tactex,urieli2016mdp} uses an MDP and dynamic programming based strategy derived from \citeauthor{Tesauro}'s bidding strategy to predict bid prices, which is the motivation for our MDP-based strategy. \citeauthor{chowdhury2016predicting} \citeyear{chowdhury2016predicting} predicts bid prices for the wholesale market PDAs using REPTree, Linear Regression and NN with weather data, with the former being is used in the SPOT \cite{Chowdhury2017} broker. \citeauthor{chowdhury2018bidding} \citeyear{chowdhury2018bidding} use a Monte Carlo Tree Search (MCTS) based strategy coupled with a REPTree based price predictor \cite{chowdhury2016predicting} and heuristics, to determine optimal bid prices. AgentUDE \cite{ozdemir2015agentude} uses an adaptive Q-learning based strategy in the wholesale market. None of these strategies are backed up by game theoretic analysis, where as our work is to build strategies derived from Nash Equilibrium.

\section{Theoretical Approach and Proofs}
\label{section:theory}
In this section, we focus solely on the best response and Nash Equilibrium analysis of double auctions.

\subsection{Nash Equilibrium analysis in single unit Double Auctions}
Consider a single unit double auction, with the \emph{clearing price and payment rule} given by \emph{ACPR}. To find a generic Nash Equilibrium in this setting, we first try to simplify the double auction by restricting the number of buyers and sellers and their behavior. Upon doing so, we derive the following case-wise results.

\subsubsection{One buyer and One Seller (OBOS)}
\label{section:OBOS}
Let's assume that one buyer and one seller participate in the double auction, with their types as $\theta_B$ and $\theta_S$ respectively. We assume that both deploy scaling based strategies, i.e., a bid by a buyer is $b_B = \alpha_B \theta_B$ and an ask by the seller is $b_S = \alpha_S \theta_S$ where $\alpha_B$ and $\alpha_S$ are the scale factors by which the buyer and seller scale their true types while bidding, respectively. Motivated by the literature {\cite{rothkopf1980equilibrium} \cite{vincent1995bidding} \cite{narahari2014game}}, we choose scale based bidding strategies for this Nash Equilibrium analysis, as compared to additive bidding strategies.

We assume $\theta_B\sim U[l_B, h_B]$ and $\theta_S \sim U[l_S,h_S]$ and this is common knwoledge. We also assume Equation \eqref{eqn:OBOSassumption}, which states that the buyer's bid (seller's ask) at any point will be less (higher) than or equal to the highest (lowest) possible seller's ask (buyer's bid).
{
\begin{equation}
	\frac{\alpha_B}{\alpha_S}\theta_B \leq h_S, \qquad \frac{\alpha_S}{\alpha_B}\theta_S \geq l_B
  	\label{eqn:OBOSassumption}
\end{equation}
}
% \sab{THESE ASSUMPTIONS ARE TO HOLD FOR ALL VALUES OF SCALE FACTORS? OR ONLY THOSE SCALE FACTORS SATISFYING THESE ARE TO BE CONSIDERED? THIS PART IS NOT CLEAR. ALSO PHYSICAL SIGNIFICANCE IS NOT CLEAR.}
Thus, the utility of the buyer if its bid gets cleared, is denoted by the difference of true valuation and clearing price. Given the true types are picked over a distribution, the expected utility is computed as:
{
    \begin{equation}
    \begin{aligned}
    u_B &= \int_{l_S}^{\frac{\alpha_B}{\alpha_S}\theta_B} \big[ \theta_B - \big( \frac{\alpha_B \theta_B + \alpha_S \theta_S}{2} \big) \big] d\theta_S \\
    & = \int_{l_S}^{\frac{\alpha_B}{\alpha_S}\theta_B} \big[ \big( 1 - \frac{\alpha_B}{2} \big) \theta_B - \frac{\alpha_S \theta_S}{2} \big] d\theta_S \\
    & = \theta_B \big( 1 - \frac{\alpha_B}{2} \big) \big( \frac{\alpha_B}{\alpha_S}\theta_B - l_S \big) - \frac{\alpha_S}{4} \big[\big( \frac{\alpha_B}{\alpha_S}\theta_B\big)^{2} - l_S^{2} \big] \\
    \end{aligned}
    \end{equation}
}
Now assuming that the buyer decides to fix its $\alpha_B$ before even seeing its own type, then its utility is given by:
{
\begin{equation}
\begin{aligned}
U_B &= \int_{l_B}^{h_B} u_B d\theta_B \\
&= \int_{l_B}^{h_B} \bigg[ \theta_B \big( 1 - \frac{\alpha_B}{2} \big) \big( \frac{\alpha_B}{\alpha_S}\theta_B - l_S \big) - \frac{\alpha_S}{4} \big[\big( \frac{\alpha_B}{\alpha_S}\theta_B\big)^{2} - l_S^{2} \big] \bigg] d\theta_B \\
% & = \frac{\alpha_B}{\alpha_S}\big( 1 - \frac{\alpha_B}{2} \big) \big( \frac{h_B^3 - l_B^3}{3} \big) - l_S \big( 1 - \frac{\alpha_B}{2} \big) \big( \frac{h_B^2 - l_B^2}{2} \big) - \frac{\alpha_B^2}{4\alpha_S}\big( \frac{h_B^3 - l_B^3}{3} \big) + \frac{\alpha_S}{4}l_S^2\big(h_B - l_B \big)\\
% & = \big( \frac{h_B^3 - l_B^3}{3} \big) \big(\frac{\alpha_B}{\alpha_S} - \frac{\alpha_B^2}{2\alpha_S} - \frac{\alpha_B^2}{4\alpha_S} \big)- l_S \big( 1 - \frac{\alpha_B}{2} \big)\big( \frac{h_B^2 - l_B^2}{2} \big) + \frac{\alpha_S}{4}l_S^2\big(h_B - l_B \big)\\
& = \big( \frac{h_B^3 - l_B^3}{3} \big) \big(\frac{\alpha_B}{\alpha_S} - \frac{3\alpha_B^2}{4\alpha_S} \big)- l_S \big( 1 - \frac{\alpha_B}{2} \big)\big( \frac{h_B^2 - l_B^2}{2} \big) + \frac{\alpha_S}{4}l_S^2\big(h_B - l_B \big)\\
\end{aligned}
\end{equation}
}
Now, differentiating w.r.t. $\alpha_B$ and equating to 0 to find maxima:
{
\begin{equation}
\begin{aligned}
\frac{\partial U_B}{\partial \alpha_B} &= 0\\
  \Rightarrow\quad
  & \big( \frac{h_B^3 - l_B^3}{3} \big) \big(\frac{1}{\alpha_S} - \frac{3\alpha_B}{2\alpha_S} \big) + l_S\big( \frac{h_B^2 - l_B^2}{4} \big)
  = 0 \\
%   &\Rightarrow\quad
%   &\big( \frac{h_B^3 - l_B^3}{3} \big) \big(\frac{3\alpha_B}{2\alpha_S} - \frac{1}{\alpha_S} \big)
%   &= l_S\big( \frac{h_B^2 - l_B^2}{4} \big) \\
%   &\Rightarrow\quad
%   &\frac{3\alpha_B}{2} - 1
%   &= \frac{3\alpha_S l_S}{4}\big( \frac{h_B^2 - l_B^2}{h_B^3 - l_B^3} \big) \\
  \Rightarrow\quad
  & \alpha_B
  = \frac{2}{3} + \frac{\alpha_S l_S}{2}\big( \frac{h_B^2 - l_B^2}{h_B^3 - l_B^3} \big)
%   \Rightarrow\quad & \alpha_B= \frac{4}{3} \bigg( \frac{2 + l_S x}{4 - l_S h_B x y} \bigg)
\end{aligned}
\label{eqn:OBOS-alphaB-Pre}
\end{equation}
}

Similarly, for the seller, the utility comes out to be:
{
% \fontsize{6.5}{8}
\begin{equation}
\begin{aligned}
u_S &= \int_{\frac{\alpha_S}{\alpha_B}\theta_S}^{h_B} \big[ \big( \frac{\alpha_B \theta_B + \alpha_S \theta_S}{2} \big) - \theta_S \big] d\theta_B \\
& = \int_{\frac{\alpha_S}{\alpha_B}\theta_S}^{h_B} \big[\frac{\alpha_B}{2} \theta_B + (\frac{\alpha_S }{2} - 1)\theta_S \big] d\theta_B \\
& = \frac{\alpha_B}{4} \big( h_B^2 - (\frac{\alpha_S}{\alpha_B}\theta_S)^2 \big) + \theta_S (\frac{\alpha_S }{2} - 1) \big( h_B - \frac{\alpha_S}{\alpha_B}\theta_S \big) \\
\end{aligned}
\end{equation}
}

Again, assuming that the seller decides to fix its $\alpha_S$ before even seeing its own type, then its utility is given by:
{
% \fontsize{6.5}{8}
\begin{equation}
\begin{aligned}
U_S &= \int_{l_S}^{h_S} u_S d\theta_S\\
% & = \int_{l_S}^{h_S} \big[ \frac{\alpha_B}{4} \big( h_B^2 - (\frac{\alpha_S}{\alpha_B}\theta_S)^2 \big) + \theta_S (\frac{\alpha_S }{2} - 1) \big( h_B - \frac{\alpha_S}{\alpha_B}\theta_S \big) \big] d\theta_S\\
& = \frac{\alpha_B h_B^2}{4} (h_S^2 - l_S^2) - \frac{\alpha_S^2}{12\alpha_B^2}(h_S^3 - l_S^3) + h_B(\frac{\alpha_S }{2} - 1)(\frac{h_S^2 - l_S^2}{2}) - (\frac{\alpha_S }{2} - 1)\frac{\alpha_S}{\alpha_B}(\frac{h_S^3 - l_S^3}{3})
\end{aligned}
\end{equation}
}
Now, differentiating w.r.t $\alpha_S$ and equating to 0 to find maxima
{
% \fontsize{6.5}{8}
\begin{equation}
\begin{aligned}
\frac{\partial U_S}{\partial \alpha_S} &= 0\\
  \Rightarrow\quad
  &(\frac{h_S^3 - l_S^3}{3}) \big[-\frac{\alpha_S}{2\alpha_B} + \frac{1}{\alpha_B} - \frac{\alpha_S}{\alpha_B} \big] + h_B (\frac{h_S^2 - l_S^2}{4}) = 0 \\
%   &\Rightarrow\quad
%   &(\frac{h_S^3 - l_S^3}{3}) \big[\frac{3\alpha_S}{2\alpha_B} - \frac{1}{\alpha_B} \big] 
%   & = h_B (\frac{h_S^2 - l_S^2}{4})\\
%   &\Rightarrow\quad
%   &\frac{3\alpha_S}{2} - 1 
%   & =  \frac{3\alpha_B h_B}{4} (\frac{h_S^2 - l_S^2}{h_S^3 - l_S^3})\\
  \Rightarrow\quad
  &\alpha_S = \frac{2}{3} + \frac{\alpha_B h_B}{2} (\frac{h_S^2 - l_S^2}{h_S^3 - l_S^3})\\
\end{aligned}
\end{equation}
}
Next, simplifying the expressions for  $\alpha_B$ and $\alpha_S$ by letting $\frac{h_B^2 - l_B^2}{h_B^3 - l_B^3}$ = $x$ and $\frac{h_S^2 - l_S^2}{h_S^3 - l_S^3}$ = $y$, we get
{
%   \fontsize{6.5}{8}
  \begin{equation}
    \begin{aligned}
        &\alpha_S = \frac{2}{3} + \frac{\alpha_B h_B y}{2}\\
      \Rightarrow\quad
        &\alpha_S = \frac{2}{3} + \frac{h_B y}{2}(\frac{2}{3} + \frac{\alpha_S l_S x}{2})\\
      \Rightarrow\quad
        &\alpha_S = \frac{4}{3} \bigg( \frac{2 + h_B y}{4 - l_S h_B x y} \bigg)
    \end{aligned}
    \label{eqn:OBOS-alphaS}
  \end{equation}
  \begin{equation}
    \begin{aligned}
        &\alpha_B = \frac{2}{3} + \frac{\alpha_S l_S x}{2}\\
      \Rightarrow\quad
        &\alpha_B = \frac{2}{3} + \frac{l_S x}{2}\frac{4}{3}\bigg( \frac{2 + h_B y}{4 - l_S h_B x y} \bigg)\\
      \Rightarrow\quad
        &\alpha_B = \frac{4}{3} \bigg( \frac{2 + l_S x}{4 - l_S h_B x y} \bigg)
    \end{aligned}
    \label{eqn:OBOS-alphaB}
  \end{equation}
}

Putting $l_S$ = $l_B$ = $0$ and $h_S$ = $h_B$ = $1$ in Equations \eqref{eqn:OBOS-alphaS} and \eqref{eqn:OBOS-alphaB}, we get $\alpha_S$ = $1$ and $\alpha_B$ = $\frac{2}{3}$. The above discussion is summarized as the following theorem.

\begin{theorem}
\label{OBOSThm}
For a single unit double auction with  \emph{ACPR}, with only one buyer and one seller, whose true types are drawn from a $0-1$ uniform distribution, if they deploy scaling based bidding strategies $b_B$ and $b_S$ which satisfy Equation \eqref{eqn:OBOSassumption} and fix their scaling factors $\alpha_B$ and $\alpha_S$ before seeing their true types, then $\alpha_S$ = $1$ and $\alpha_B$ = $\frac{2}{3}$ constitute a Nash Equilibrium.
\end{theorem}
% \sab{SHOULD IT BE REITERATED HERE THAT THESE VALUES FOUND HERE FORM A NASH EQUILIBRIUM?}

\subsubsection{Two Buyers and One Seller (TBOS)}
Let's assume that two buyers B1 and B2, and one seller participate in the double auction, with types $\theta_{B1}$, $\theta_{B2}$ and $\theta_{S}$ respectively. We assume that all deploy scaling based strategies, and both buyers have the same scaling factor $\alpha_B$. Thus, a bid by buyer B1 is $b_{B1} = \alpha_B \theta_{B1}$ and by buyer B2 is $b_{B2} = \alpha_B \theta_{B2}$, while a bid by the seller is $b_S = \alpha_S \theta_{S}$. We also assume Equation \eqref{eqn:TBOSassumption}, which states that the first buyer's (seller's) bid at any point will be less than or equal to the highest possible seller's (buyer's) bid.
{
% \fontsize{6.5}{8}
\begin{equation}
	\frac{\alpha_{B}}{\alpha_S}\theta_{B1} \leq h_S, \qquad \frac{\alpha_{S}}{\alpha_B}\theta_S \leq h_B
    \label{eqn:TBOSassumption}
\end{equation}
}
First, we find the utility of the first buyer. We consider the following cases:

\begin{enumerate}
  \item $b_{B1} \geq b_{B2} \Rightarrow \theta_{B1}\geq \theta_{B2}$ and  $b_{B2} \geq b_S$ $\Rightarrow$ $\theta_{S} \leq \frac{\alpha_B}{\alpha_S}\theta_{B2}$
  
  Let the utility in this case be denoted by $u_{b11}$.
{
	\fontsize{9}{10}
    \begin{equation}
        \begin{aligned}
        u_{b11} &= \int_{l_B}^{\theta_{B1}} \bigg[ \int_{l_S}^{\frac{\alpha_B}{\alpha_S}\theta_{B2}} \big[ \theta_{B1} - \big( \frac{\alpha_B \theta_{B2} + \alpha_S \theta_S}{2} \big) \big] d\theta_S \bigg] d\theta_{B2} \\
%         & = \int_{l_B}^{\theta_{B1}} \bigg[ \theta_{B1} \big(\frac{\alpha_B}{\alpha_S}\theta_{B2} - l_S \big) - \frac{\alpha_B \theta_{B2}}{2}\big(\frac{\alpha_B}{\alpha_S}\theta_{B2} - l_S \big) - \frac{\alpha_S}{4}\big((\frac{\alpha_B}{\alpha_S}\theta_{B2})^2 - (l_S)^2 \big) \bigg] d\theta_{B2} \\
        & = \int_{l_B}^{\theta_{B1}} \bigg[ - \frac{3\alpha_B^2}{4\alpha_S} \theta_{B2}^2  + (\frac{\alpha_B}{\alpha_S}\theta_{B1} + \frac{\alpha_B l_S}{2})\theta_{B2} + (-l_S \theta_{B1} + \frac{l_S^2 \alpha_S}{4})  \bigg] d\theta_{B2} \\
%         & = - \frac{\alpha_B^2}{4\alpha_S} (\theta_{B1}^3 - l_B^3)  + (\frac{\alpha_B}{\alpha_S}\theta_{B1} + \frac{\alpha_B l_S}{2})(\frac{\theta_{B1}^2 - l_B^2}{2}) + (-l_S \theta_{B1} + \frac{l_S^2 \alpha_S}{4})(\theta_{B1} - l_B) \\
        & = \theta_{B1}^3 (-\frac{\alpha_B^2}{4\alpha_S} + \frac{\alpha_B}{2\alpha_S}) + \theta_{B1}^2 (\frac{\alpha_B l_S}{4} - l_S) + \theta_{B1} (-\frac{\alpha_B l_B^2}{2\alpha_S} + l_S l_B + \frac{l_S^2 \alpha_S}{4}) + (\frac{\alpha_B^2 l_B^3}{4\alpha_S} - \frac{\alpha_B l_S l_B^2}{4} - \frac{\alpha_S l_B l_S^2}{4})
        \end{aligned}
    \end{equation}
}

  \item $b_{B1} \geq b_{S} \Rightarrow \theta_{B1}\geq \frac{\alpha_S}{\alpha_B}\theta_{S}$  and $b_{B2} \leq b_S \Rightarrow \theta_{B2} \leq \frac{\alpha_S}{\alpha_B}\theta_{S}$
  
  Let the utility in this case be denoted by $u_{b12}$.
{
	\fontsize{9}{10}
    \begin{equation}
        \begin{aligned}
        u_{b12} &= \int_{l_S}^{\frac{\alpha_B}{\alpha_S}\theta_{B1}} \bigg[ \int_{l_B}^{\frac{\alpha_S}{\alpha_B}\theta_{S}} \big[ \theta_{B1} - \big( \frac{\alpha_B \theta_{B1} + \alpha_S \theta_S}{2} \big) \big] d\theta_{B2}  \bigg] d\theta_S\\
        &= \int_{l_S}^{\frac{\alpha_B}{\alpha_S}\theta_{B1}} \bigg[ \theta_{B1} (1 - \frac{\alpha_B}{2})(\frac{\alpha_S}{\alpha_B}\theta_{S} - l_B)- \frac{\alpha_S}{2}(\frac{\alpha_S}{\alpha_B}\theta_{S} - l_B)\theta_S \bigg] d\theta_S\\
%         &= -\frac{\alpha_S^2}{6\alpha_B}(\frac{\alpha_B^3}{\alpha_S^3}\theta_{B1}^3 - l_S^3) + (\frac{l_B \alpha_S}{4} + \frac{\alpha_S}{2\alpha_B}\theta_{B1} - \frac{\alpha_S}{4}\theta_{B1})(\frac{\alpha_B^2}{\alpha_S^2}\theta_{B1}^2 - l_S^2) - \theta_{B1} l_B (1 - \frac{\alpha_B}{2})(\frac{\alpha_B}{\alpha_S}\theta_{B1} - l_S)\\
%         &= \theta_{B1}^3(-\frac{\alpha_B^2}{6\alpha_S} + \frac{\alpha_B}{2\alpha_S} - \frac{\alpha_B^2}{4\alpha_S}) + \theta_{B1}^2 l_B(\frac{\alpha_B^2}{4\alpha_S} + \frac{\alpha_B^2}{2\alpha_S} - \frac{\alpha_B}{\alpha_S}) + \theta_{B1} ( - \frac{\alpha_S l_S^2}{2\alpha_B} + \frac{\alpha_S l_S^2}{4} + l_B l_S - \frac{\alpha_B l_B l_S}{2}) + (\frac{\alpha_S^2 l_S^3}{6} - \frac{\alpha_S l_B l_S^2}{4}) \\
        &= \theta_{B1}^3(-\frac{5\alpha_B^2}{12\alpha_S} + \frac{\alpha_B}{2\alpha_S}) + \theta_{B1}^2 l_B(\frac{3\alpha_B^2}{4\alpha_S} - \frac{\alpha_B}{\alpha_S}) + \theta_{B1} ( - \frac{\alpha_S l_S^2}{2\alpha_B} + \frac{\alpha_S l_S^2}{4} + l_B l_S - \frac{\alpha_B l_B l_S}{2}) + (\frac{\alpha_S^2 l_S^3}{6} - \frac{\alpha_S l_B l_S^2}{4})
        \end{aligned}
    \end{equation}
}
\end{enumerate}

Now assuming that the first buyer decides to fix its $\alpha_B$ before even seeing its own type, then we find the utility to be - 
{
% 	\fontsize{6.5}{8}
    \begin{equation}
      \begin{aligned}
        U_{B1} &= \int_{l_B}^{h_B} u_{B1} d\theta_{B1} = \int_{l_B}^{h_B} (u_{b11} + u_{b12}) d\theta_{B1}\\
        &= \int_{l_B}^{h_B} \bigg[ \theta_{B1}^3(-\frac{2\alpha_B^2}{3\alpha_S} + \frac{\alpha_B}{\alpha_S})  + \theta_{B1}^2 (\frac{\alpha_B l_S}{4}  - \frac{\alpha_B l_B}{\alpha_S} +\frac{3\alpha_B^2 l_B}{4 \alpha_S} - l_S) \\
        &+ \theta_{B1} \bigg(2 l_B l_S -\frac{\alpha_B l_B^2}{2\alpha_S} - \frac{\alpha_S l_S^2}{2\alpha_B} + \frac{\alpha_S l_S^2}{2} - \frac{\alpha_B l_B l_S}{2}\bigg) + (\frac{\alpha_B^2 l_B^3}{4\alpha_S} - \frac{\alpha_B l_S l_B^2}{4} - \frac{\alpha_S l_B l_S^2}{2 } + \frac{\alpha_S^2 l_S^3}{6} ) \bigg] d\theta_{B1}\\
      \end{aligned}
    \end{equation}
}
Now, differentiating w.r.t $\alpha_B$ and equating to 0 to find maxima
{
    % \fontsize{6.5}{8}
    \begin{equation}
        \begin{aligned}
            &\frac{\partial U_{B1}}{\partial \alpha_B} = 0\\
%             \Rightarrow\quad 
%             &\bigg[ (\frac{h_B^4 - l_B^4}{4})(\frac{\alpha_B}{6\alpha_S}) + (\frac{h_B^3 - l_B^3}{3}) (- \frac{l_S}{4} + \frac{h_B - l_B}{\alpha_S} -\frac{3\alpha_B}{2 \alpha_S} (h_B - l_B)) \\
%             & + (\frac{h_B^2 - l_B^2}{2}) \bigg(-\frac{l_B^2}{2\alpha_S} + \frac{l_S h_B}{2} + \frac{\alpha_S l_S^2}{2\alpha_B^2} - \frac{l_B l_S}{2} \bigg) + (h_B - l_B) \big( \frac{\alpha_B l_B^3}{2\alpha_S} - \frac{l_S l_B^2}{4}\big) \bigg]
%             = 0 \\
           	\Rightarrow\quad 
            &\bigg[ (h_B - l_B) \big( \frac{\alpha_B l_B^3}{2\alpha_S} - \frac{l_S l_B^2}{4}\big) + (\frac{h_B^2 - l_B^2}{2}) \bigg(-\frac{l_B^2}{2\alpha_S} + \frac{\alpha_S l_S^2}{2\alpha_B^2} - \frac{l_B l_S}{2} \bigg)\\
            & + (\frac{h_B^3 - l_B^3}{3}) \bigg(\frac{l_S}{4} - \frac{l_B}{\alpha_S} +\frac{3\alpha_B l_B}{2 \alpha_S}\bigg) + (\frac{h_B^4 - l_B^4}{4})(-\frac{4\alpha_B}{3\alpha_S} + \frac{1}{\alpha_S})\bigg]
            = 0 \\
            \Rightarrow\quad 
            &\bigg[ \frac{\alpha_B}{6\alpha_S}(-2h_B^3 + 4l_B^3 + l_B^2 h_B - 2l_Bh_B^2) + \frac{\alpha_S}{4\alpha_B^2}(h_B l_S^2 + l_B l_S^2) \\
            &+ \frac{l_S(h_B^2 -2 l_B^2 - 2h_Bl_B)}{12} + \frac{3h_B^3 - 4l_B^3 - 4h_B l_B^2 - l_B h_B^2}{12 \alpha_S} \bigg]
            = 0 \\
            %\Rightarrow\quad 
            %&\bigg[ \frac{\alpha_B^3}{24}(-11h_B^3 + 25l_B^3 + l_B^2 h_B + l_Bh_B^2) + %\frac{\alpha_S^2}{4}(h_B l_S^2 + l_B l_S^2) \\
            %&+ \frac{\alpha_S \alpha_B^2 l_S(2h_B^2 -7 l_B^2 - h_Bl_B)}{12} + %\frac{\alpha_B^2(4h_B^3 - 7l_B^3 - 3h_B l_B^2)}{12} \bigg]
            %= 0 \\
        \end{aligned}
        \label{eqn:TBOS-Buyer1}
    \end{equation}
}
Similarly, for the seller we find the utility. We again have 4 cases:

\begin{enumerate}
    \item $b_{B1} \geq b_{B2} \Rightarrow \theta_{B1}\geq \theta_{B2}$ and  $b_{B2} \geq b_S$ $\Rightarrow$ $\theta_{B2} \geq \frac{\alpha_S}{\alpha_B}\theta_{S}$
  
    Let the utility in this case be denoted by $u_{s1}$.
{
% 	\fontsize{6.5}{8}
    \begin{equation}
        \begin{aligned}
        u_{s1} &= \int_{\frac{\alpha_S}{\alpha_B}\theta_{S}}^{h_B} \bigg[ \int_{\theta_{B2}}^{h_B} \big[ \big( \frac{\alpha_B \theta_{B2} + \alpha_S \theta_S}{2} \big) - \theta_{S} \big] d\theta_{B1} \bigg] d\theta_{B2}\\
        &= \int_{\frac{\alpha_S}{\alpha_B}\theta_{S}}^{h_B} \bigg[ \theta_S (\frac{\alpha_S}{2} - 1)(h_B - \theta_{B2}) + \frac{\alpha_B}{2} \theta_{B2}(h_B - \theta_{B2}) \bigg] d\theta_{B2}\\
%         &= -\frac{\alpha_B}{6}(h_B^3 - \frac{\alpha_S^3}{\alpha_B^3}\theta_S^3) + (\frac{\theta_S}{2} - \frac{\alpha_S}{4}\theta_S + \frac{\alpha_B h_B}{4})(h_B^2 - \frac{\alpha_S^2}{\alpha_B^2}\theta_S^2) + \theta_S h_B (\frac{\alpha_S}{2} - 1)(h_B - \frac{\alpha_S}{\alpha_B}\theta_S)\\
        &= \theta_S^3 (\frac{5\alpha_S^3}{12\alpha_B^2} - \frac{\alpha_S^2}{2\alpha_B^2}) + \theta_S^2 (-\frac{3\alpha_S^2 h_B}{4\alpha_B} + \frac{\alpha_S h_B}{\alpha_B}) + \theta_S (\frac{\alpha_S}{4} - \frac{1}{2})h_B^2 + (\frac{\alpha_B h_B^3}{12})
        \end{aligned}
    \end{equation}
}
    \item $b_{B2} \geq b_{B1}\Rightarrow \theta_{B2}\geq \theta_{B1}$ and $b_{B1} \geq b_S \Rightarrow \theta_{B1} \geq \frac{\alpha_S}{\alpha_B}\theta_{S}$
  
    Let the utility in this case be denoted by $u_{s2}$. Since the two buyers are symmetric, the utility in this case comes to be same as in case 1.
{
% 	\fontsize{6.5}{8}
    \begin{equation}
        \begin{aligned}
        u_{s2} &= \theta_S^3 (\frac{5\alpha_S^3}{12\alpha_B^2} - \frac{\alpha_S^2}{2\alpha_B^2}) + \theta_S^2 (-\frac{3\alpha_S^2 h_B}{4\alpha_B} + \frac{\alpha_S h_B}{\alpha_B}) + \theta_S (\frac{\alpha_S}{4} - \frac{1}{2})h_B^2 + (\frac{\alpha_B h_B^3}{12})
        \end{aligned}
    \end{equation}
}
    \item $b_{B1} \geq b_S$ $\Rightarrow \theta_{B1}\geq \frac{\alpha_S}{\alpha_B}\theta_{S}$ and $b_{B2} \leq b_S \Rightarrow \theta_{B2} \leq \frac{\alpha_S}{\alpha_B}\theta_{S}$
  
    Let the utility in this case be denoted by $u_{s3}$.
{
% 	\fontsize{6.5}{8}
    \begin{equation}
        \begin{aligned}
        u_{s3} &= \int_{l_B}^{\frac{\alpha_S}{\alpha_B}\theta_{S}} \bigg[ \int_{\frac{\alpha_S}{\alpha_B}\theta_{S}}^{h_B} \big[ \big( \frac{\alpha_B \theta_{B1} + \alpha_S \theta_S}{2} \big) - \theta_{S} \big] d\theta_{B1} \bigg] d\theta_{B2}\\
        &= \int_{l_B}^{\frac{\alpha_S}{\alpha_B}\theta_{S}} \bigg[ \theta_S(\frac{\alpha_S}{2} - 1)(h_B - \theta_S) + \frac{\alpha_B}{4}(h_B^2 - \theta_S^2)\bigg] d\theta_{B2}\\
%         & = \bigg[ \theta_S(\frac{\alpha_S}{2} - 1)(h_B - \frac{\alpha_S}{\alpha_B}\theta_S) + \frac{\alpha_B}{4}(h_B^2 - \frac{\alpha_S^2}{\alpha_B^2}\theta_S^2)\bigg] (\frac{\alpha_S}{\alpha_B}\theta_{S} - l_B)\\
        &= \theta_S^3 (\frac{\alpha_S^2}{\alpha_B^2} - \frac{3\alpha_S^3}{4\alpha_B^2}) + \theta_S^2 (\frac{\alpha_S^2 h_B }{2 \alpha_B} + \frac{3\alpha_S^2 l_B}{4\alpha_B} - \frac{\alpha_S h_B}{\alpha_B} - \frac{\alpha_S l_B}{\alpha_B}) + \theta_S (\frac{\alpha_S h_B^2}{4} - h_B l_B(\frac{\alpha_S}{2} - 1)) - \frac{\alpha_B h_B^2 l_B}{4}\\
        \end{aligned}
    \end{equation}
}
    \item $b_{B2} \geq b_S$ $\Rightarrow \theta_{B2}\geq \frac{\alpha_S}{\alpha_B}\theta_{S}$ and $b_{B1} \leq b_S \Rightarrow \theta_{B1} \leq \frac{\alpha_S}{\alpha_B}\theta_{S}$
  
    Let the utility in this case be denoted by $u_{s4}$. Since the two buyers are symmetric, the utility in this case comes to be same as in case 3.
{
% 	\fontsize{6.5}{8}
    \begin{equation}
        \begin{aligned}
        u_{s4} &= \theta_S^3 (\frac{\alpha_S^2}{\alpha_B^2} - \frac{3\alpha_S^3}{4\alpha_B^2}) + \theta_S^2 (\frac{\alpha_S^2 h_B }{2 \alpha_B} + \frac{3\alpha_S^2 l_B}{4\alpha_B} - \frac{\alpha_S h_B}{\alpha_B} - \frac{\alpha_S l_B}{\alpha_B}) + \theta_S (\frac{\alpha_S h_B^2}{4} - h_B l_B(\frac{\alpha_S}{2} - 1)) - \frac{\alpha_B h_B^2 l_B}{4}\\
        \end{aligned}
    \end{equation}
}
\end{enumerate}

Now assuming that the seller decides to fix its $\alpha_S$ before even seeing its own type, then we find the utility to be - 
{
% 	\fontsize{6.5}{8}
    \begin{equation}
        \begin{aligned}
            U_{S} &= \int_{l_S}^{h_S} u_{S} d\theta_{S} \\
            &= \int_{l_S}^{h_S} (u_{s1} + u_{s2} + u_{s3} + u_{s4}) d\theta_{S} = 2 \int_{l_S}^{h_S} (u_{s1} + u_{s3}) d\theta_{S}\\
            &= 2\int_{l_S}^{h_S} \bigg[ \theta_S^3 ( \frac{\alpha_S^2}{2\alpha_B^2} - \frac{\alpha_S^3}{3\alpha_B^2}) + \theta_S^2 (\frac{3\alpha_S^2 l_B}{4\alpha_B} - \frac{\alpha_S l_B}{\alpha_B} -\frac{\alpha_S^2 h_B}{4\alpha_B} ) \\
            &+ \theta_S (\frac{\alpha_S h_B^2}{2} - h_B l_B(\frac{\alpha_S}{2} - 1)  - \frac{h_B^2}{2})+ \frac{\alpha_B h_B^3}{12} - \frac{\alpha_B h_B^2 l_B}{4}\bigg] d\theta_{S} \\
        \end{aligned}
    \end{equation}
}

Now, differentiating w.r.t $\alpha_S$ and equating to 0 to find maxima
{
% 	\fontsize{6.5}{8}
    \begin{equation}
        \begin{aligned}
        &\frac{\partial U_{S}}{\partial \alpha_S} = 0\\
            \Rightarrow\quad 
            &\bigg[ (\frac{h_S^4 - l_S^4}{4}) ( \frac{\alpha_S}{\alpha_B^2} - \frac{\alpha_S^2}{\alpha_B^2}) + (\frac{h_S^3 - l_S^3}{3}) (\frac{3\alpha_S l_B}{2\alpha_B} - \frac{l_B}{\alpha_B} -\frac{\alpha_S h_B}{2\alpha_B} ) + (\frac{h_S^2 - l_S^2}{2}) (-\frac{h_B l_B}{2} + \frac{h_B^2}{2}) \bigg]
            = 0 \\
            \Rightarrow\quad 
            &-\frac{\alpha_S^2}{4\alpha_B^2}(h_S^4 - l_S^4) + \frac{\alpha_S}{\alpha_B}( \frac{h_S^4 - l_S^4}{4\alpha_B} + \frac{(h_S^3 -l_S^3)(3l_B - h_B)}{6}) \\
            &- \frac{l_B}{\alpha_B}(\frac{h_S^3 -l_S^3}{3})+ (\frac{h_S^2 - l_S^2}{2}) (-\frac{h_B l_B}{2} + \frac{h_B^2}{2})
            = 0 \\
        \end{aligned}
        \label{eqn:TBOS-Seller}
    \end{equation}
}

From Equation {\eqref{eqn:TBOS-Buyer1}}, we have a bi-variate cubic equation in $\alpha_B$ and $\alpha_S$, and from Equation {\eqref{eqn:TBOS-Seller}}, we have a bi-variate quadratic equation in $\alpha_B$ and $\alpha_S$. 

Assuming $\alpha_S \neq 0$ and $\alpha_B \neq 0$ (non-zero bids), and putting $l_S$ = $l_B$ = $0$ and $h_S$ = $h_B$ = $1$ in Equation {\eqref{eqn:TBOS-Buyer1}}, we get

{
% \fontsize{6.5}{8}
\begin{equation}
    \begin{aligned}
        &\frac{-4\alpha_B}{3\alpha_S} + \frac{1}{\alpha_S} = 0\\
        \Rightarrow\quad
        &\alpha_B = \frac{3}{4}\\
    \end{aligned}
    \label{eqn:TBOS-Buyer-Solved}
\end{equation}
}

Now, putting $\alpha_B = \frac{3}{4}$ (from Equation {\eqref{eqn:TBOS-Buyer-Solved}}), $l_S$ = $l_B$ = $0$ and $h_S$ = $h_B$ = $1$ in Equation {\eqref{eqn:TBOS-Seller}}, we get

{
% \fontsize{6.5}{8}
\begin{equation}
    \begin{aligned}
        & -\alpha_S^2 + \frac{\alpha_S}{2} + \frac{9}{16} = 0\\
        \Rightarrow\quad
        &\alpha_S = \frac{1 \pm \sqrt{10}}{4}\\
    \end{aligned}
    \label{eqn:TBOS-Seller-Solved}
\end{equation}
}

Since $\alpha_S$ = $\frac{1 - \sqrt{10}}{4} < 0$ (negative scaling factor), we ignore this solution.

Thus, putting $l_S$ = $l_B$ = $0$ and $h_S$ = $h_B$ = $1$ in Equation \eqref{eqn:TBOS-Buyer1} and Equation \eqref{eqn:TBOS-Seller}, we get $\alpha_S$ = $\frac{1 + \sqrt{10}}{4} \approx 1.0406$ and $\alpha_B$ = $\frac{3}{4} = 0.75$.

The above discussion can be summarized as the following theorem.
\begin{theorem}
\label{TBOSThm}
For a single unit double auction with \emph{ACPR} with two buyers and one seller, whose true types are drawn from a $0-1$ uniform distribution, if they deploy scaling based strategies $b_{B1}$, $b_{B2}$ and $b_S$, with buyers having the same scaling factor $\alpha_B$, which satisfy Equation \eqref{eqn:TBOSassumption} and fix their scaling factors $\alpha_B$ and $\alpha_S$ before seeing their true types, then $\alpha_S = \frac{1 + \sqrt{10}}{4}$ and $\alpha_B = \frac{3}{4}$ constitute a Nash Equilibrium. 
\end{theorem}

As seen, with the increase in just one buyer, the complexity of the solution increases. It becomes increasingly difficult to extend and generalize the above results for a realistic market setting. Thus, moving forward, taking the PowerTAC wholesale market as testbed, we present a bidding strategy and experimentally show that it follows the theoretical results obtained in this section.

\subsection{Best Response analysis in multi-unit Double Auctions with complete information}

In practice, there are key differences between double auctions implemented in markets, and the theoretical results arrived above, stated as follows:

\begin{enumerate}
    
    \item Quantity may be involved in the trading market auctions, which is not considered above.

    \item The seller needs to use the same bidding strategy for one to achieve the above result, which may not the case.
\end{enumerate}

So, considering a multi-unit double auction with $ACPR$, where bids are of the form $(quantity, price)$. Let $Q_a$ denote the quantity not cleared of the Last Cleared Ask if it is executed partially, and let $Q_b$ denote the quantity not cleared of the Last Cleared Bid if it is executed partially.

\begin{claim}
Upon clearance of an auction, either $Q_a$ or $Q_b$, or both have to be zero.
\end{claim}

\begin{proof}
If the last bid partially clears, $Q_a$ $=$ $0$ and $Q_b$ $\neq$ $0$, and if the last ask partially clears, $Q_a$ $\neq$ $0$ and $Q_b$ $=$ $0$. If both clear fully, $Q_a$ $=$ $0$ and $Q_b$ $=$ $0$. The last bid and last ask both can not clear partially, as, if they did, then more quantity can be cleared with last bid's price higher than the last ask's price.
\end{proof}
% \addtocounter{theorem}{7}

Next, we propose the best response if all the other bids are known to the bidder (i.e. complete information).
\begin{proposition}
\label{prop:LCP}
When a buyer (seller) has complete information about the auction, and it desires to procure (sell) entire energy it bids (asks) for, it's a best response to bid as close as possible to the last clearing bid (ask).
\end{proposition}

\begin{proof}

Let $b_i(pb_{i}, qb_{i})$ denote the $i^{th}$ bid placed in the auction, for $qb_{i}$ amount of energy, at $pb_{i}$ price. Similarly, let $a_i(pa_{i}, qa_{i})$ denote the $i^{th}$ ask placed in the auction, where $pa_{i}$ and $qa_{i}$ denote the asking price and quantity respectively. For simplicity, let us assume the ordering of bids to be in descending order of price, and ordering of asks to be in ascending order of price. Therefore, $pb_{i} > pb_{i+1}$, and $pa_{i}$ $<$ $pa_{i+1}$. The last clearing bid (LCB) is denoted by $b_{c_1}(pb_{c_1}, qb_{c_1})$, while the last clearing ask (LCA) is denoted by $a_{c_2}(pa_{c_2}, qa_{c_2})$. Thus, by ACPR, the clearing price is given as $CP = (pb_{c_1} + pa_{c_2})/2$. Let $Q_a$ denote the energy not cleared of the LCA if LCA is executed partially, and let $Q_b$ denote the energy not cleared of the LCB if LCB is executed partially.

WLOG, let us consider the case of bids in the auction. Our claim essentially solves the optimization problem of minimizing the clearing price while procuring the full amount of energy. Assume a buyer $m$ wants to place a bid $b_m(pb_m, qb_m)$ in such an auction. We define $Q_a$ and $Q_b$ denote the energy not cleared of last bid and last ask respectively, when the buyer $m$ doesn't participate. Now if the buyer does participate in the auction, there are the following possibilities (depicted in Figure \ref{fig:cases}):

\begin{figure}[!ht]
  \centering
  \includegraphics[]{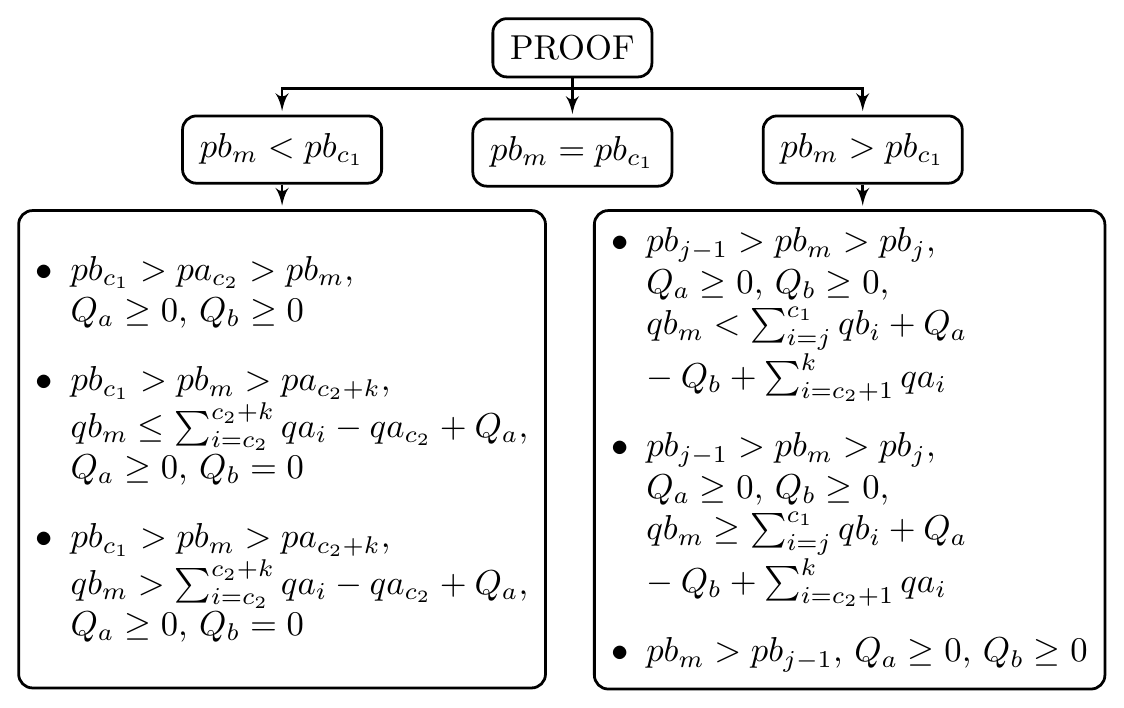}
  \caption{Proof Cases}
  \label{fig:cases}
\end{figure}
\begin{case}
$pb_m$ $<$ $pb_{c_1}$ i.e. bid price is lower than price of the would-be cleared bid when the buyer doesn't participate

\begin{enumerate}
\item $pb_{c_1} > pa_{c_2} > pb_m$ and $Q_a \geq 0$, $Q_b \geq 0$ i.e. if the bidding price of $m$ is lower than the ask price of the last cleared ask (when buyer doesn't participate). Under this condition, the bid doesn't clear, and this is clearly not optimal, as the buyer doesn't get it's required energy.

\item $pb_{c_1} > pb_m > pa_{c_2+k}$, $Q_a \geq 0$ and $Q_b = 0$ and $k$ is the smallest index with $k \geq 0$ such that $qb_m \leq \sum_{i=c_2}^{c_2+k}qa_i - qa_{c_2} + Q_a$, i.e. if the last cleared bid is fully executed when buyer doesn't participate, and buyer $m$'s bid price is higher than the next closest ask. In this case, clearly, $b_m$ becomes the last clearing bid, and it gets cleared fully. Thus, buyer $m$ gets energy at the lowest price possible by having the last cleared bid.

\item $pb_{c_1} > pb_m > pa_{c_2+k}$, where $Q_a \geq 0$ and $Q_b$ = $0$ and $k$ is the largest index with $k \geq 0$ such that $qb_m > \sum_{i=c_2}^{c_2+k}qa_i - qa_{c_2} + Q_a$, i.e. if the last cleared bid is fully executed when buyer doesn't participate, and buyer $m$'s bid price is higher than the next ask. In this case, clearly, $b_m$ becomes the last clearing bid, and it gets cleared partially. Although the buyer $m$ gets some energy at the lowest price possible by having the last cleared bid, it would've been better off bidding higher than the previous bid $b_c$, in order to clear it's entire bid energy $qb_m$.

% \item $pb_{c_1} > pb_m > pa_{c_2}$ and $Q_a$ $\neq$ $0$ i.e. if the last cleared ask (when broker doesn't participate) is partially executed, and the bidding price of the broker is higher than the ask price of the last cleared ask (when broker doesn't participate), the bid $b_m$ is cleared. In this case, a maximum of $Q_a$ amount of energy is cleared from bid $b_m$ and $b_m$ becomes the clearing bid. If $Q_a < qb_m$, the broker would've been better off bidding just higher than the next higher bid $b_{c1}$, as it would've got the full energy $qb_m$ in that case, instead of getting a max of $Q_a$ amount of energy, which supports our claim.
\end{enumerate}
\end{case} 

\begin{case}
$pb_m$ $=$ $pb_{c_1}$ i.e. bidding price same as the last cleared bid's price (when buyer doesn't participate). This is a probability zero event, and extremely unlikely to occur. Since it's a tie, it'll either be treated as Case 1 or Case 3, depending on the tiebreaker rule set by the auction.
\end{case}

\begin{case}
$pb_m$ $>$ $pb_{c_1}$ i.e. bidding price is just higher than the last cleared bid's price (when buyer doesn't participate)

\begin{enumerate}

\item $pb_{j-1} > pb_m > pb_j$ and $Q_a \geq 0$, $Q_b \geq 0$, where $j$ is the largest index with $j \leq c_1$ such that $qb_m \leq \sum_{i = j}^{c_1} qb_i + Q_a - Q_b + \sum_{i = c_2 + 1}^{k}{qa_i}$ and $k$ is the largest index such that $pa_k$ $\leq$ $pb_{j}$. In this case, $b_m$ clears fully and becomes the second last bid to clear, and $b_j$ clears partially, or $b_m$ clears fully and becomes the last cleared bid. If the buyer decides to bid below $pb_j$, it'll clear partially, which is not desirable, and thus supports our claim.

\item $pb_{j-1} > pb_m > pb_j$ and $Q_a \geq 0$, $Q_b \geq 0$, where $j$ is the smallest index such that $qb_m > \sum_{i = j}^{c_1} qb_i + Q_a - Q_b + \sum_{i = c_2 + 1}^{k}{qa_i}$ and $k$ is the largest index such that $pa_k \leq pb_{j}$. In this case, $b_m$ becomes the last clearing bid, and executes partially. The buyer is better off bidding higher than $b_{j-1}$ as it would've cleared fully.

\item $pb_m > pb_{j-1}$ and $Q_a \geq 0$, $Q_b \geq 0$, where $j$ is the largest index such that $qb_m \leq \sum_{i = j}^{c_1} qb_i + Q_a - Q_b + \sum_{i = c_2 + 1}^{k}{qa_i}$ and $k$ is the largest index such that $pa_k$ $\leq$ $pb_{j}$. In this case, $b_m$ clears fully. But, the buyer $m$ would get the full amount of energy even if it bid between $pb_{j-1}$ and $pb_j$, and would've then become close to the clearing bid.

\end{enumerate}
\end{case}
\end{proof}

Given the above proposition in a complete information setting, we further propose a MDP-based bidding strategy, which uses the past auction trends and statistics, to achieve the best response with incomplete information in the PowerTAC wholesale market.

\section{MDPLCPBS: PowerTAC Wholesale Market Bidding Strategy}
\label{section:MDPLCPBS}

\sloppy We introduce the MDP and LCP based Bidding Strategy (MDPLCPBS) for the PowerTAC wholesale market. The PowerTAC wholesale market accepts bids of the form $(energy\text{ }amount, limit\text{-}price)$. With respect to a broker, let the energy amount being sold be positive, while the energy amount being bought be negative. Meanwhile, let negative price indicate a broker is earning revenue, while positive price indicate it is paying or losing revenue. Thus, from the viewpoint of a broker, a buy order is seen to have a negative $energy\text{ }amount$ and a positive $limit\text{-}price$, while a sell order (termed as an \emph{ask}), is seen to have a positive $energy\text{ }amount$ and a negative $limit\text{-}price$.

At timeslot $t$, assuming a broker has a predicted demand profile $D_t$ $=$ $\{d_{t + 1}, d_{t+2}, \ldots, d_{t+24}\}$, where $d_i$ is the predicted net demand at timeslot $i$. Also, let $P_t$ $=$ $\{p_{t+1}, p_{t+2}, \ldots, p_{t+24}\}$ denote the amount of energy already procured by past energy contracts, where $p_i$ denotes the energy procured for timeslot $i$. Thus, the remaining energy to be procured is given by $E_t$ $=$ $\{e_{t + 1}, e_{t + 2},\ldots, e_{t+24}\}$, where $e_i$ $=$ $d_i - p_i$ is the net energy left to be procured for timeslot $i$. The bidding strategy, MDPLCPBS, to procure the aforementioned energy requirements, comprises of three major submodules - (i) Limit Price Predictor, (ii) Quantity Predictor, and (iii) Last Cleared Price Predictor.

\subsection{Limit Price Predictor (LPP)} At any given timeslot $t$, the predictor computes 24 $limit\text{-}prices$ for 24 simultaneous PDAs in the PowerTAC wholesale market. Motivated by \cite{Tesauro} and \cite{urieli2014tactex}, the Limit Price Predictor uses the following MDP to place optimal $limit\text{-}prices$ for bids:

\begin{enumerate}
	\item \textbf{States}: $s \in S = \{0,1,\dots,24,success\}$, $s_0 := 24$
    
    \item \textbf{Actions}: $limit\text{-}price$ $\in \mathbb{R}$
    
    \item \textbf{Transition}: The same state transition from \cite{urieli2014tactex} is used. A state $s \in \{1,\ldots, 24\}$ transitions to one of two states. If a bid is partially or fully cleared, it transitions to the terminal state $success$. Otherwise, a state $s$ transitions to state $s - 1$. The clearing (i.e. transition) probability $p_{cleared}(s, limit\text{-}price)$ is initially unknown and is determined by Equation (\ref{eqn:pcleared}).
    
    \item \textbf{Reward}: At any state $s \in \{1,\dots,24\}$, the reward is $0$. At terminal state $s=0$, the reward is the negative of the \emph{balancing price} per unit energy. At terminal state $s = success$, the reward is the negative of the \emph{limit-price} of the cleared bid. Since we take the price to be positive for bids and negative for asks, maximizing reward results in minimizing costs.
    
    \item \textbf{Terminal States}: $\{0, success\}$
\end{enumerate}

We solve the above MDP using a sequential bidding strategy, that computes the optimal bid $limit\text{-}price$ that minimizes the expected procurement cost per unit energy. It uses the $balancing\text{-}price$ as the expected price at state $s = 0$, and recursively minimizes the expected cost by using the probability of clearance, $p_{cleared}(s, limit\text{-}price)$. This solution is summarized as a value function, stated as follows:
\begin{equation}
        V(s) = \begin{cases}
                    balancing\text{-}price, & \text{if}\ s = 0 \\
                    \displaystyle\min_{limit\text{-}price}\{p_{cleared} \times limit\text{-}price\\
                     + (1 - p_{cleared}) \times V(s-1)\}, & \text{if}\ s \in [1,24]
                \end{cases}
                \label{eqn:dprog}
\end{equation}
Given that the $balancing\text{-}price$ and the $p_{cleared}$ values are different for bids and asks, we maintain two separate instances of the MDP, and solve them independently.

The value function in Equation (\ref{eqn:dprog}) is solved recursively using dynamic programming. However, before doing so, the $balancing\text{-}price$ and the transition function $p_{cleared}(s, limit\text{-}price)$ need to be estimated, as they are both initially unknown. The $balancing\text{-}price$ is estimated by averaging the balancing-prices across past timeslots. On the other hand, the clearing probability, $p_{cleared}(s, limit\text{-}price)$, is computed using past auction statistics as:
{
\begin{equation}
        p_{cleared} = \dfrac{\sum_{ac \in auction[s], ac.\text{LCP} < limit\text{-}price} ac.\text{cleared-amount}}{\sum_{ac \in auction[s]} ac.\text{cleared-amount}} \label{eqn:pcleared}
    \end{equation}
}
where $auction[s]$ is the set of all past auctions in the state $s$, and LCP is the \emph{Last Clearing Price}, which is estimated by the \emph{Last Cleared Price Predictor}. The auction statistics for each state $s$ are re-used in the future for estimating $p_{cleared}$, as we iterate over the same sequence of states S during the bidding process.

\subsection{Quantity Predictor (QP)}
The \emph{Quantity Predictor} is primarily responsible for distributing the demand for a target timeslot across all the 24 auctions, in order to further reduce overall energy cost. The idea is to buy more and sell less at cheaper prices, and vice-versa. It essentially breaks down the demand for a target timeslot $t + 24$, across auctions in timeslots $\{t, t + 1, \ldots, t + 23\}$.

For each auction state $s \in \{1,\dots,24\}$ at timeslot $t$, it takes the corresponding energy requirement $e_{t + s}$ and uses the 24 \emph{limit-prices} from the \emph{Limit Price Predictor} to distribute the required energy. The energy quantity to bid/ask, for each state $s$ at timeslot $t$, is given by:
{
	\begin{equation}
        q(s) = \begin{cases}
                    \frac{e_{t+s}}{\sum_{j = s}^{24} \frac{limit\text{-}price[j]}{limit\text{-}price[s]}}, & \text{if}\ e_{t+s} > 0 \\
                    \frac{e_{t+s}}{\sum_{j = s}^{24} \frac{limit\text{-}price[s]}{limit\text{-}price[j]}}, & \text{if}\ e_{t+s} < 0 \\
                    0, & \text{if}\ e_{t+s} = 0
                \end{cases} \label{eqn:cases}
    \end{equation}
}
where $s \in \{1,\dots,24\}$, $limit\text{-}price[s]$ is the limit-price for state $s$ determined by the \emph{Limit Price Predictor}. The first case in Equation \eqref{eqn:cases} refers to the situation where energy needs to sold, so the bid quantity is directly proportional to the predicted limit-price of that auction - essentially selling more energy at higher price. On the other hand, the second case occurs when the energy needs to be procured. So, the bid quantity is set to be inversely proportional to the predicted limit-price i.e. buying more energy at cheaper price. Thus, the final bid is of the form $(q(s), limit\text{-}price[s])$.

\subsection{Last Cleared Price Predictor (LCPP)}
\label{section:LCPP}
First, one has to note that, in any auction, the LCP is greater than or equal to CP. Mostly, $LCP > CP$, as $P(LCP = CP) = 0$, i.e. LCP equal to CP is a probability zero event. In PowerTAC, the LCP is not known to any broker. In essence, one can place better bids if the LCP for each auction is known, as they can bid higher than a predicted LCP to become the last bid, and achieve best response according to Proposition \ref{prop:LCP}. The Last Cleared Price Predictor essentially tries to determine the LCP for bids and asks for all executed auctions. It does so by probing the auctions with a set of dummy orders, which have the minimum tradeable energy as quantity (0.01 MwH), and $limit\text{-}prices$ equally spaced in the range $[\beta \times limit\text{-}price, balancing\text{-}price]$. After execution, the LCP for bids for an auction in state $s$ is determined by:
{
    \begin{equation}
        LCP(s) = \min(dummy\text{-}bids_{cleared}, limit\text{-}price[s]_{cleared})
        \label{eqn:lcpbid}
    \end{equation}
}
where $dummy\text{-}bids_{cleared}$ is the set of bid prices of all dummy bids which got cleared in the state $s$, and $limit\text{-}price[s]_{cleared}$ is the limit-price for the cleared final bid made in state $s$ (taken to be infinity if final bid did not clear or does not exist). Similarly, the LCP for asks is given as:
{
    \begin{equation}
        LCP(s) = \max(dummy\text{-}asks_{cleared}, limit\text{-}price[s]_{cleared})
        \label{eqn:lcpask}
    \end{equation}
}
where $dummy\text{-}asks_{cleared}$ is the set of ask prices of all dummy asks which got cleared in the state $s$, and $limit\text{-}price[s]_{cleared}$ is the limit-price for the cleared final ask made in state $s$ (taken to be infinity if final ask did not clear or does not exist). These LCP values are then used to update the clearing probability $p_{cleared}$ in Equation \eqref{eqn:pcleared}.
\begin{algorithm}[!ht]
	\caption{MDPLCPBS}\label{MDPLCPBS}
	{
% 	\fontsize{8}{10}
    \begin{algorithmic}[1]
    \Procedure{MDPLCPBS}{$energyReq[1..24]$}
       	\State $marketData[0..24] \gets getMarketStatistics()$
       	\If {$EnoughDataPoints(marketData)$}
       		\State {\scriptsize $bidPrices[1..24] \gets SolveMDP(marketData)$}
            \State {\scriptsize $bidQty[1..24] \gets SpreadQty (energyReq, bidPrices)$}
       	\Else
       		\State $bidPrices[1..24] \gets SampleBiddingPolicy()$
            \State $bidQty[1..24] \gets energyReq[1..24]$
       	\EndIf
        \State $sendBids(bidPrices, bidQty)$
        \State $sendDummyBids(bidPrices, marketData)$
    \EndProcedure
    \end{algorithmic}
    }
\end{algorithm}
Algorithm \ref{MDPLCPBS} summarizes MDPLCPBS, which is executed every timeslot. It takes the energy requirement for the 24 auctions as input. First it collects the market statistics, which includes the LCP estimate and clearing amount from previous timeslots, and the balancing price (line 2). If the number of data points is suitable enough (taken to be 24 in our implementation), it proceeds to solve the MDP and generates a set of prices to bid (line 4). Using these set of prices, and the energy requirements, it generates a set of quantities to bid (line 5). If data points are not enough, the bidding policy given in the PowerTAC \emph{sample-broker} is used to determine the bid prices (line 7), and the bid quantities are set as the full energy requirements (line 8). Using the determined bid prices and quantities, we place the actual bids (line 10), and a set of dummy bids in the market (line 11). The time complexity of Algorithm \ref{MDPLCPBS} comes out to be in the order of the number of past market data points.

\section{Experimental Analysis}
\begin{figure*}[!ht]
  \begin{subfigure}{.24\columnwidth}
    \centering
    % include first image
    \includegraphics[width=\linewidth]{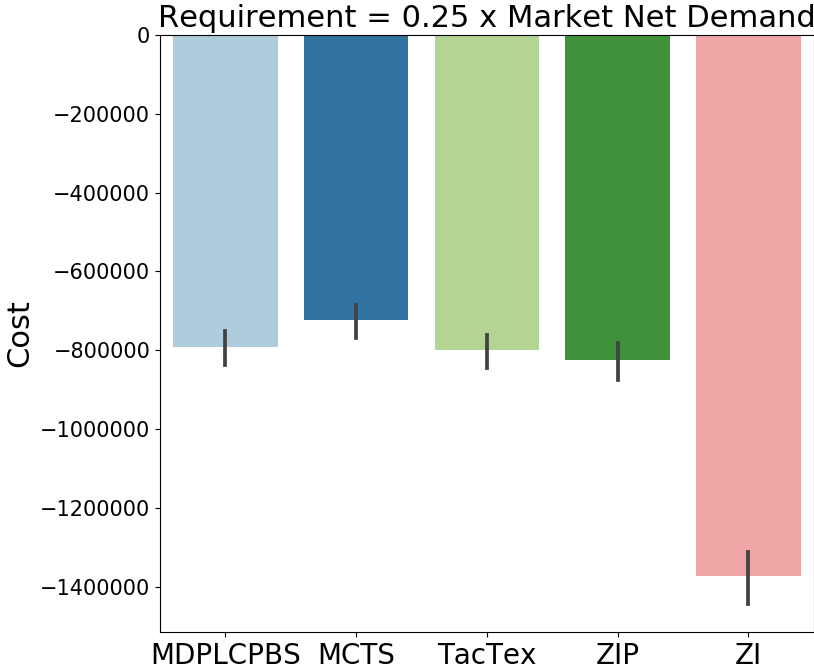}  
    \caption{}
    \label{fig:ND0.25}
  \end{subfigure}
  \begin{subfigure}{.24\columnwidth}
    \centering
    % include second image
    \includegraphics[width=\linewidth]{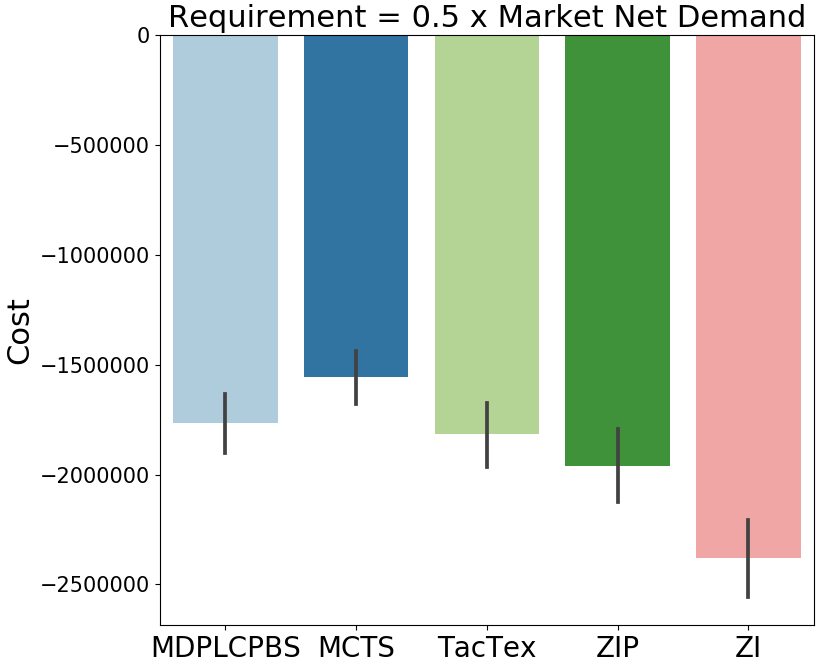}  
    \caption{}
    \label{fig:ND0.5}
  \end{subfigure}
  \begin{subfigure}{.24\columnwidth}
    \centering
    % include second image
    \includegraphics[width=\linewidth]{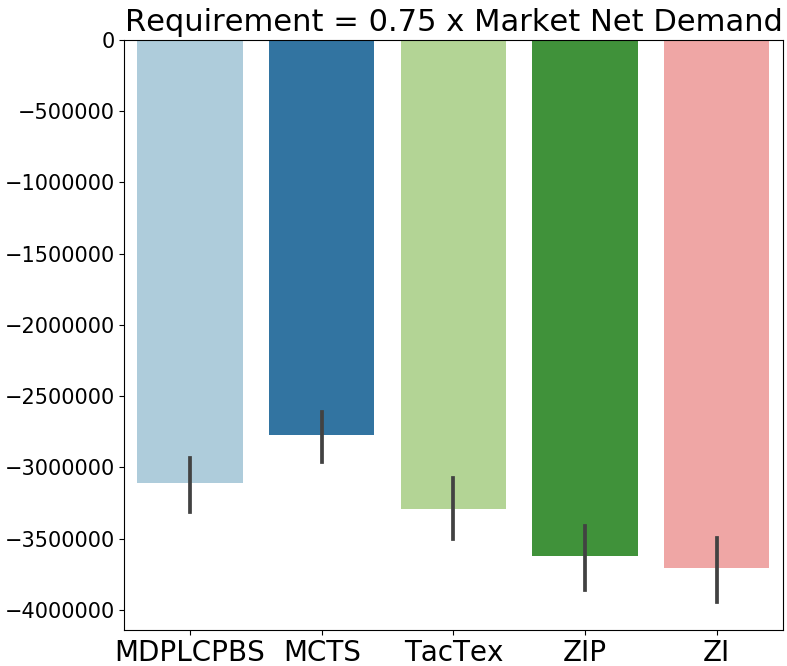}
    \caption{}
    \label{fig:ND0.75}
  \end{subfigure}
  \begin{subfigure}{.24\columnwidth}
    \centering
    % include second image
    \includegraphics[width=\linewidth]{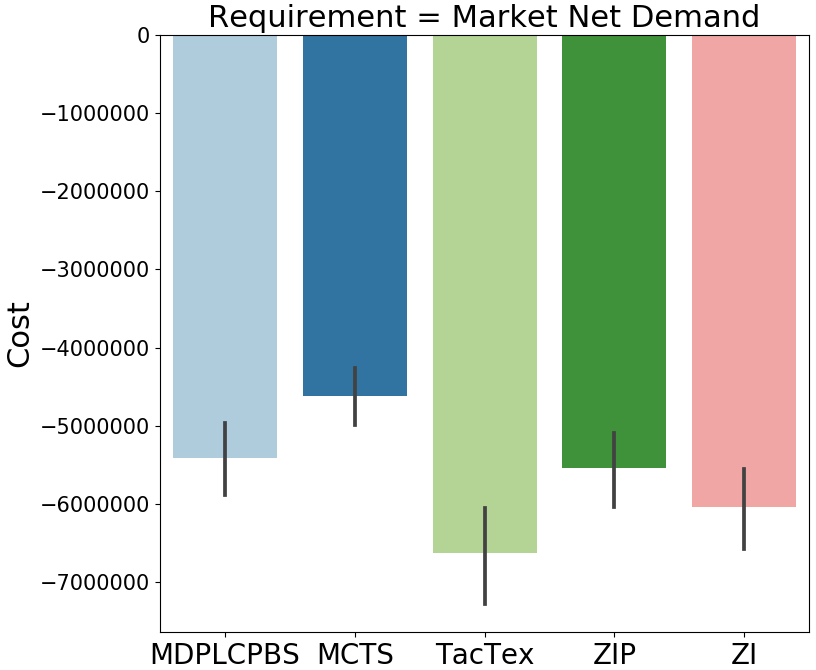}  
    \caption{}
    \label{fig:ND}
  \end{subfigure}
  \caption{Net cost comparison of strategies across games with different energy requirements}
  \label{fig:benchmarks}
\end{figure*}

We first analyze if our proposed bidding strategy, MDPLCPBS, follows the Nash Equilibrium derived above, and then benchmark it against the baseline and competing state-of-the-art strategies.

\subsection{Validation Experiments}

% \sg{In table 1,2 it is OSBS shd b clear. what is buyers theoretical value shd be mentioned in the table. text can b then minimized}

\begin{table*}[ht!]
    \centering
    \resizebox{\textwidth}{!}{
    \begin{tabular}{||c|c|c|c|c|c||c|c|c|c|c|c||} 
        \hline
         & \multicolumn{5}{|c||}{Fixed seller's scale factor and corresponding buyer's scale factor}
         & \multicolumn{5}{c||}{Fixed buyer's scale factor and corresponding seller's scale factor}\\
        \hline
        % \multicolumn{1}{||p{2cm}|}{\centering Statistic\\ Scaling Factor}
        Scaling Factor
        & \multirow{2}{*}{\shortstack{0.948689 \\ (-0.1)}}
        & \multirow{2}{*}{\shortstack{0.998689 \\ (-0.05)}}
        & \multirow{2}{*}{\shortstack{\textbf{1.048689} \\ (Theoretical Value)}}
        & \multirow{2}{*}{\shortstack{ 1.098689 \\ (+0.05)}}
        & \multirow{2}{*}{\shortstack{ 1.148689 \\ (+0.1)}}
        & \multirow{2}{*}{\shortstack{ 0.791386 \\ (-0.1)}} 
        & \multirow{2}{*}{\shortstack{0.841386 \\ (-0.05)}} 
        & \multirow{2}{*}{\shortstack{\textbf{0.891386} \\ (Theoretical Value)}}
        & \multirow{2}{*}{\shortstack{ 0.941386 \\ (+0.05)}}
        & \multirow{2}{*}{\shortstack{0.991386 \\ (+0.1)}} \\ [0.5ex]
        \cline{1-1}
        Statistic & & & & & & & & & & \\
        \hline \hline
        Average & 0.772435 & 0.804782 & 0.838087 & 0.863553 & 0.907389 & 0.989438 & 1.057427 & 1.113121 & 1.226423 & 1.616557\\ [0.5ex]
        \hline
        \multicolumn{1}{||p{2cm}|}{\centering Std. Dev.} & 0.033287 & 0.037697 & 0.025127 & 0.020749 & 0.036637 & 0.033287 & 0.037697 & 0.025127 & 0.020749 & 0.036637\\ [0.5ex]
        \hline
    \end{tabular}}
    \caption{OBOS - Experimental scale factors values for buyer and seller using MDPLCPBS}
    \label{table:OBOSscalefactors}
\end{table*}

We take the Power TAC simulator and isolate the wholesale market, and remove all market simulator participants (GenCos, internal buyers) from the market. We test the one buyer one seller (OBOS) scenario by deploying only two agents in the isolated wholesale market - a buyer and a seller. These agents have a fixed energy demand that they need to buy (sell) from the market. In these experiments, we set the energy demand to be the previous slot's tariff market net demand, which both the buyer and seller are notified about. We draw $\theta_B\sim U[40, 80]$ and $\theta_S \sim U[40,80]$, and compute the theoretical scale factors using Equation {\eqref{eqn:OBOS-alphaB}} and Equation {\eqref{eqn:OBOS-alphaS}}. We run two batches of experiments, with 30 games in each set of the batch, for 5 sets per batch. During each batch, one of the agents has a fixed scaling based bidding strategy, while the other uses MDPLCPBS.

In the first (second) batch, we draw the seller's (buyer's) valuation $\theta_S \sim U[40,80]$ ($\theta_B\sim U[40, 80]$), and apply a fixed scale factor within $\pm0.1$ of the theoretical value. The buyer (seller) generates its valuation $\theta_B\sim U[40, 80]$ ($\theta_B\sim U[40, 80]$), and uses this valuation as the $balancing\text{-}price$ in MDPLCPBS to generate bids.

The experimental average scale factor and standard deviation, for cleared bids, for the buyer and the seller in the two batches of experiments are documented in Table {\ref{table:OBOSscalefactors}}. The table demonstrates that in a one buyer and one seller setting, MDPLCPBS approaches the Nash Equilibrium characterized for a single unit OBOS double auction. The values demonstrate that as the fixed scale factor for the seller is increased, the buyer's scale factor increases slowly. On the contrary, when the fixed scale factor for the buyer is increased, the seller's scale factor increases rapidly. 

\subsection{Benchmarks}
\begin{table}[b!]
    \centering
    % \resizebox{\columnwidth}{!}{
    \begin{tabular}{||c|c|c|c|c|c|c||} 
        \hline
        \multirow{2}{*}{\shortstack{\% of Market \\ Demand}}
        & \multirow{2}{*}{Statistic}
        & \multicolumn{5}{|c||}{State}\\
        \cline{3-7}
        & & 24 & 23 & 22 & 21 & 20 \\ [0.5ex]
        \hline\hline
        \multirow{3}{*}{100}
        & Wt. Avg. Relative Error (\%) & 9.35 & 10.4 & 7.59 & 6.66 & 8.34\\
        \cline{2-7}
        & Std. Dev. of \% Error & 13.47 & 26.76 & 21.00 & 21.38 & 44.71\\
        \cline{2-7}
        & Avg. Cleared Quantity & 885.76 & 34.87 & 23.08 & 17.17 & 14.36\\
        \hline
        % \multirow{3}{*}{75}
        % & Relative Error (\%) & 10.34 & 11.83 & 10.41 & 12.14 & 19.78\\
        % \cline{2-7}
        % & Std. Dev. of \% Error & 14.39 & 28.05 & 24.08 & 29.12 & 94.84\\
        % \cline{2-7}
        % & Avg. Traded Quantity & 885.76 & 34.87 & 23.08 & 17.17 & 14.36\\
        % \hline
        \multirow{3}{*}{50}
        & Wt. Avg. Relative Error (\%) & 15.18 & 24.03 & 10.28 & 10.96 & 16.36\\
        \cline{2-7}
        & Std. Dev. of \% Error & 20.79 & 46.31 & 33.35 & 32.82 & 53.51\\
        \cline{2-7}
        & Avg. Cleared Quantity & 836.82 & 22.19 & 13.54 & 8.89 & 6.92\\
        \hline
        % \multirow{3}{*}{25}
        % & Relative Error (\%) & 10.34 & 11.83 & 10.41 & 12.14 & 19.78\\
        % \cline{2-7}
        % & Std. Dev. of \% Error & 14.39 & 28.05 & 24.08 & 29.12 & 94.84\\
        % \cline{2-7}
        % & Avg. Traded Quantity & 885.76 & 34.87 & 23.08 & 17.17 & 14.36\\
        % \hline
    \end{tabular}
    % }
    \caption{Weighted relative error rate for LCP prediction}
    \label{table:lcperror}
\end{table}
We isolate the PowerTAC wholesale market from the full PowerTAC simulator while keeping the market simulator participants (GenCos, internal buyers) and weather simulator, and benchmark the performance of MDPLCPBS. The following agents/brokers are used in these benchmarks:
\begin{itemize}
\item Zero Intelligence (ZI): The ZI agent \cite{gode1993allocative} uses a randomized bid strategy and ignores the market state. It generates random order prices, ignoring the state of the market. In our experiments, we derive its bids from a uniform distribution with mean $\mu$ and a standard deviation of \$10. The mean $\mu$ taken from the limit price predicted by the MDP in TacTex \cite{urieli2014tactex}. The broker places one bid per auction, and the remaining required energy as the bid quantity. It continues to do the same for all the 24 bidding opportunities, or until the required energy is procured.
\item Zero Intelligence Plus (ZIP): The ZIP agent \cite{tesauro2001high} maintains a scalar variable $m$ denoting its desired profit margin, and it combines this with a unit's limit price to compute a bid price $p$. For each failed trade, the price is adjusted by small increments to beat the failed bid price $p$. In our experiments, the initial limit price value $\mu$ is determined from the limit price predicted by the MDP in TacTex. The profit margin $m$ is set to $1\%$ of $\mu$, resulting in the initial bid price to be $p = \mu \times 1.01$. If the bid fails, the next bid price is incremented by $10\%$ of $\mu$. Then, the new bid price is given by $p = \mu \times 1.11$.
\item TacTex: The TacTex \cite{urieli2014tactex} agent uses an MDP based model and dynamic programming to determine limit-prices for bids. The algorithm described in the paper was implemented and used in our experiments.
\item MCTS: The MCTS \cite{chowdhury2018bidding} agent uses a Monte Carlo Tree Search (MCTS) coupled with heuristics on top of the limit price derived from a REPTree based limit price predictor, to determine the optimal bid price. In our experiments, we used the MCTS-dyn-C2 version with 10000 iterations, which is shown to be the best performing variation of the MCTS bidding strategy.
\end{itemize}

For a timeslot $t + 24$ in the future, having 24 bidding opportunities in timeslots $\{t, t+1, \ldots, t+23\}$, the energy to be procured is set to be same across all the brokers. This energy amount for $t + 24$ is determined as some fraction of the net demand in timeslot $t$ in the PowerTAC simulation tariff market. Four sets of 10 games each are simulated, with each set having a different fraction of the net demand to be procured. The fraction set is given by $\{0.25, 0.5, 0.75, 1\}$.

Figure \ref{fig:benchmarks} shows the net cost of all the agents across the four sets of games. In each case, MDPLCPBS outperforms ZI, ZIP and TacTex on a consistent basis, while losing out to MCTS. While ZIP may seem to perform reasonably well in some cases, it can be countered easily in strategic settings, like in single-shot single-unit auction setting, whereas MDPLCPBS follows the equilibrium. It is also to be noted that, while MCTS uses tailored heuristics, MDPLCPBS is derived from the game theoretic analysis of single shot double auction (Proposition \ref{prop:LCP}). We leave the game theoretic analysis of MCTS for future work.

Table \ref{table:lcperror} summarizes the weighted relative error rates (weighed by cleared quantity) in predicting the LCP. 82\% of the total cleared energy for a future timeslot is cleared in the first auction itself, with 91\% of the total being cleared in the first five. Since the cleared energy of the other states is extremely low, their corresponding predicted LCPs have less impact on the $p_{cleared}$ calculation. Thus, we focus on the error rates of the states corresponding to the first five auctions for a target timeslot. We see that MDPLCPBS has a 10\% and 15\% error rate in the LCP prediction for the first auction of a timeslot (state 24), for 100\% and 50\% of the market demand as requirement. Moreover, the corresponding weighted average relative error across all states, comes out to be 12\% and 18\% respectively. The error rates increase as requirement decrease, as there are more low quantity bids by brokers and Miso buyer's low bids often set the LCP. Thus, MDPLCPBS predicts the LCP with minimal error, during auctions where most of the energy gets traded.

% -43.6248849/-47.3447883

\section{Conclusion}
% \sg{dont start with MDPLCPBS..start with PDAs thm 6,7 then cm t MDPLCPBS and expxerimental validation}
% \sgh{Usage of words: derived (most proofs in supplementary)}
In this paper, we first analytically characterized Nash Equilibrium strategies for a single unit double auction with the clearing price and payment rule as \emph{ACPR}, for OBOS, and TBOS with scale based bidding strategies. We also proposed the best response in a complete information setting in a multi-unit double auction with \emph{ACPR}. Based on these formulations, we presented MDPLCPBS, a bidding strategy for PDAs. Furthermore, we experimentally validated that MDPLCPBS achieves the Nash Equilibrium derived for single unit double auction with \emph{ACPR} for OBOS. Finally, we benchmarked MDPLCPBS against the baseline and competing state-of-the-art strategies, and showed that it outperforms most of them consistently. Simultaneously, we showed that it predicts the LCP with minimal error.

% \fontsize{9.5}{10.5}
% \selectfont
\bibliography{references}
\bibliographystyle{named}

\end{document}